\documentclass{article}
\usepackage{spconf,amsmath,graphicx,xargs}

\newcommand{\calN}{\mathcal{N}}
\newcommand{\calO}{\mathcal{O}}
\newcommand{\calP}{\mathcal{P}}

\newcommand{\bfD}{\mathbf{D}}

\newcommand{\bfI}{\mathbf{I}}

\newcommand{\bfy}{\mathbf{y}}
\newcommand{\bfz}{\mathbf{z}}
\newcommand{\bfm}{\mathbf{m}}

\newcommand{\bfnot}{\mathbf{0}}

\def\msa{\mathsf{A}}

\def\msb{\mathsf{B}}

\def\msy{\mathsf{Y}}

\newcommand{\mcb}[1]{\mathcal{B}(#1)}

\def\mcp{\mathcal{P}}

\def\rset{\mathbb{R}}

\def\zset{\mathbb{Z}}
\def\nset{\mathbb{N}}
\def\nsets{\mathbb{N}^*}

\def\rmd{\mathrm{d}}

\def\rme{\mathrm{e}}

\newcommand{\abs}[1]{\left\vert #1 \right\vert}

\newcommand{\ps}[2]{\left\langle#1,#2 \right\rangle}
\newcommandx{\norm}[2][1=]{\ifthenelse{\equal{#1}{}}{\left\Vert #2 \right\Vert}{\left\Vert #2 \right\Vert^{#1}}}

\newcommand{\defEns}[1]{\left\lbrace #1 \right\rbrace }
\newcommand{\plusinfty}{+\infty}
\def\eqsp{\;}

\newcommandx\sequence[3][2=,3=]
{\ifthenelse{\equal{#3}{}}{\ensuremath{\{ #1_{#2}\}}}{\ensuremath{\{ #1_{#2}, \eqsp #2 \in #3 \}}}}
\newcommandx\sequenceD[3][2=,3=]
{\ifthenelse{\equal{#3}{}}{\ensuremath{\{ #1_{#2}\}}}{\ensuremath{( #1)_{ #2 \in #3} }}}
\newcommandx{\sequencek}[2][2=k\in\nset]{\ensuremath{( #1)_{#2}}}

\newcommand{\wrt}{w.r.t.}
\def\iid{i.i.d.}

\newcommand{\ie}{\textit{i.e.}}

\def\Id{\operatorname{Id}}

\newcommand{\ensemble}[2]{\left\{#1\,:\eqsp #2\right\}}

\newcommand{\set}[2]{\ensemble{#1}{#2}}

\def\vareps{\varepsilon}
\def\sphere{\mathbb{S}}
\def\sphereD{\mathbb{S}^{d-1}}
\newcommand{\1}{\mathbbm{1}}

\def\distance{\ell}
\newcommandx{\wasserstein}[3][1=\distance,3=]{\mathbf{W}_{#1}^{#3}\left(#2\right)}
\newcommandx{\wassersteinLigne}[3][1=\distance,3=]{\mathbf{W}_{#1}^{#3}(#2)}
\newcommandx{\wassersteinD}[1][1=\distance]{\mathbf{W}_{#1}}
\newcommandx{\wassersteinDLigne}[1][1=\distance]{\mathbf{W}_{#1}}
\newcommandx{\swasserstein}[3][1=\distance,3=]{\mathbf{SW}_{#1}^{#3}\left(#2\right)}
\newcommandx{\swassersteinLigne}[3][1=\distance,3=]{\mathbf{SW}_{#1}^{#3}(#2)}
\newcommandx{\swassersteinD}[1][1=\distance]{\mathbf{SW}_{#1}}
\newcommandx{\swassersteinDLigne}[1][1=\distance]{\mathbf{SW}_{#1}}

\def\argmin{\mathrm{argmin}}
\def\us{u}
\def\uss{u^{\star}}

\def\unifS{\boldsymbol{\sigma}}

\newcommand\fraca[2]{(#1/#2)}

\def\hmu{\hat{\mu}}
\def\hnu{\hat{\nu}}

\def\model{\mathscr{M}}

\def\noisy{\mathbf{v}}
\def\clean{\mathbf{u}}
\def\noise{\mathbf{w}}
\def\indiceset{I}
\def\indice{\mathbf{i}}
\def\indicedic{\mathbf{j}}
\def\indicetrois{\mathbf{k}}
\def\patch{P}
\def\dic{D}
\def\PSNR{\mathrm{PSNR}}

\let\zz\[\let\zzz\]
\usepackage[inline]{enumitem}
\usepackage{url}
\usepackage{subfig}
\usepackage{xcolor}
\usepackage{bbm}
\usepackage{ifthen}
\usepackage{nicefrac}
\usepackage{amssymb,mathrsfs} 
\usepackage{yfonts}
\usepackage{bbm}
\usepackage[utf8]{inputenc} 
\usepackage[T1]{fontenc}    
\usepackage{upgreek}        
\usepackage{booktabs}       
\usepackage{amsfonts}       
\usepackage{nicefrac}       
\usepackage{microtype}
\usepackage{hyperref}
\usepackage{wrapfig}
\usepackage{aliascnt}
\usepackage{cleveref}

\let\[\zz\let\]\zzz

\usepackage{amsthm}

\makeatletter

\crefname{theorem}{theorem}{Theorems}
\Crefname{Theorem}{Theorem}{Theorems}

\newaliascnt{lemma}{theorem}

\aliascntresetthe{lemma}
\crefname{lemma}{lemma}{lemmas}
\Crefname{Lemma}{Lemma}{Lemmas}

\newaliascnt{corollary}{theorem}

\aliascntresetthe{corollary}
\crefname{corollary}{corollary}{corollaries}
\Crefname{Corollary}{Corollary}{Corollaries}

\newaliascnt{proposition}{theorem}
\newtheorem{proposition}[proposition]{Proposition}
\aliascntresetthe{proposition}
\crefname{proposition}{proposition}{propositions}
\Crefname{Proposition}{Proposition}{Propositions}

\newaliascnt{definition}{theorem}

\aliascntresetthe{definition}
\crefname{definition}{definition}{definitions}
\Crefname{Definition}{Definition}{Definitions}

\newtheorem{remark}{Remark}

\crefname{example}{example}{examples}
\Crefname{Example}{Example}{Examples}

\crefname{figure}{figure}{figures}
\Crefname{Figure}{Figure}{Figures}

\usepackage{float}
\usepackage[noend]{algorithm2e}
\floatname{algorithm}{Pseudo-code}
\newenvironment{pseudocode}[1][htb]
  {%
   \begin{algorithm}[#1]%
  }{\end{algorithm}}

\title{APPROXIMATE BAYESIAN COMPUTATION WITH \\ 
THE SLICED-WASSERSTEIN DISTANCE}
\name{Kimia Nadjahi$^1$\thanks{Supported by the Data Science \& AI chair from Télécom Paris, the French National Research Agency (ANR-16-CE23-0014 project) and the Polish National Science Center grant (NCN UMO-2018/31/B/ST1/00253).}, Valentin De Bortoli$^2$, Alain Durmus$^2$, Roland Badeau$^1$, Umut \c{S}im\c{s}ekli$^1$}
\address{$^1$ LTCI, T\'el\'ecom Paris, Institut Polytechnique de Paris, France \\
$^2$ CMLA, \'Ecole normale sup\'erieure Paris-Saclay, CNRS, Université Paris-Saclay, France}

\begin{document}
\ninept
\maketitle
\begin{abstract}
	Approximate Bayesian Computation (ABC) is a popular method for approximate inference in generative models with intractable but easy-to-sample likelihood. It constructs an approximate posterior distribution by finding parameters for which the simulated data are close to the observations in terms of summary statistics. These statistics are defined beforehand and might induce a loss of information, which has been shown to deteriorate the quality of the approximation. To overcome this problem, Wasserstein-ABC has been recently proposed, and compares the datasets via the Wasserstein distance between their empirical distributions, but does not scale well to the dimension or the number of samples. We propose a new ABC technique, called Sliced-Wasserstein ABC and based on the Sliced-Wasserstein distance, which has better computational and statistical properties. We derive two theoretical results showing the asymptotical consistency of our approach, and we illustrate its advantages on synthetic data and an image denoising task.
\end{abstract}
\begin{keywords}
	Likelihood-free inference, Approximate Bayesian Computation, Optimal transport, Sliced-Wasserstein distance
\end{keywords}
\section{Introduction}
\label{sec:intro}

\vspace{-5pt}

Consider the problem of estimating the posterior distribution of some model parameters $\theta \in \rset^{d_\theta}$ given $n$ data points $\bfy_{1:n} \in \msy^n$. This distribution has a closed-form expression given by the Bayes' theorem up to a multiplicative constant:\ \ $\pi(\theta | \bfy_{1:n}) \propto \pi(\bfy_{1:n} | \theta)\pi(\theta)$.
\noindent For many statistical models of interest, the likelihood $\pi(\bfy_{1:n} | \theta)$ cannot be numerically evaluated in a reasonable amount of time, which prevents the application of classical likelihood-based approximate inference methods. Nevertheless, in various settings, even if the associated likelihood is numerically intractable, one can still generate synthetic data given any model parameter value. This generative setting gave rise to an alternative framework of likelihood-free inference techniques. Among them, Approximate Bayesian Computation (ABC, \cite{Tavare97,Beaumont2002}) methods have become a popular choice and have proven useful in various practical applications (e.g., \cite{Peters2006,Tanaka2006,Wood2010}). The core idea of ABC is to bypass calculation of the likelihood by using simulations: the exact posterior is approximated by retaining the parameter values for which the synthetic data are close enough to the observations. Closeness is usually measured with a discrepancy measure between the two datasets reduced to some `summary statistics' (e.g., empirical mean or empirical covariance). While summaries allow a practical and efficient implementation of ABC, especially in high-dimensional data spaces, the quality of the approximate posterior distribution highly depends on them and constructing sufficient statistics is a non-trivial task. Summary statistics can be designed by hand using expert knowledge, which can be tedious especially in real-world applications, or in an automated way, for instance see \cite{Fearnhead2012}. 

Recently, discrepancy measures that view data sets as empirical probability distributions to eschew the construction of summary statistics have been proposed for ABC. Examples include the Kullback-Leibler divergence (KL, \cite{Jiang18}), maximum mean discrepancy \cite{Park16}, and Wasserstein distance \cite{Bernton2019}. This latter distance emerges from the optimal transport (OT) theory and has attracted abundant attention in statistics and machine learning due to its strong theoretical properties and applications on many domains. In particular, it has the ability of making meaningful comparisons even between probability measures with non-overlapping supports, unlike KL. However, the computational complexity of the Wasserstein distance rapidly becomes a challenge when the dimension of the observations is large. Several numerical methods have been proposed during the past few years to speed-up this computation. Wasserstein-ABC (WABC, \cite{Bernton2019}) uses an approximation based on the Hilbert space-filling curve and termed the Hilbert distance, which is computationally efficient but accurate for small dimensions only. Besides, under a general setting, the Wasserstein distance suffers from a curse of dimensionality in the sense that the error made when approximating it from samples grows exponentially fast with the data space dimension \cite{Weed2019}. These computational and statistical burdens can strongly affect the performance of WABC applied to high-dimensional data. 

The Sliced-Wasserstein distance (SW, \cite{rabin:et:al:2011,bonneel2015sliced}) is an alternative OT distance and leverages the attractive property that the Wasserstein distance between one-dimensional measures has an analytical form which can be efficiently approximated. SW is defined as an average of one-dimensional Wasserstein distances, and therefore has a significantly better computational complexity. Several recent studies have reported empirical success on generative modeling with SW \cite{Deshpande2018,csimcsekli2018sliced,Kolouri2019Sliced} as well as nice asymptotic and statistical properties \cite{Bonnotte2013,deshpande2019max,2019arXiv190604516N}, making this %
alternative distance increasingly popular. In this paper, we develop a novel ABC framework that uses SW as the data discrepancy measure. This defines a likelihood-free method which does not require choosing summary statistics and is efficient even with high-dimensional observations. We derive asymptotical guarantees on the convergence of the resulting ABC posterior, and we illustrate the superior empirical performance of our methodology by applying it on a synthetical problem and an image denoising task.

\vspace{-5pt}

\section{Background}
\label{sec:background}

\vspace{-5pt}

\noindent Consider a probability space $(\Omega, \mathcal{F}, \mathbb{P})$ with associated expectation operator $\mathbb{E}$, on which all the random variables are defined. Let $(Y_k)_{k \in \nsets}$ be a sequence of independent and identically distributed (\iid) random variables associated with some observations $(\bfy_k)_{k \in \nsets}$ valued in  $\msy \subset \mathbb{R}^{d}$. Denote by $\mu_{\star}$ the common distribution of  $(Y_k)_{k \in \nsets}$ and by $\mathcal{P}(\msy)$ the set of probability measures on $\msy$. For any $n \in \nsets$, $\hmu_n = \fraca{1}{n} \sum_{i=1}^n \updelta_{Y_i}$ denotes the empirical distribution corresponding to $n$ observations. Consider a statistical model $\model_{\Theta} =  \{ \mu_\theta \in \calP(\msy)\, :\, \theta \in \Theta \}$ parametrized by $\Theta \subset \rset^{d_{\theta}}$. We focus on parameter inference for purely generative models: for any $\theta \in \Theta$, we can draw \iid~samples $(Z_{k})_{k \in \nsets}\in \msy^{\nsets}$ from $\mu_\theta$, but the numerical evaluation of the likelihood is not possible or too expensive. For any $m \in \nsets$,  $\hmu_{\theta, m} = \fraca{1}{m} \sum_{i=1}^m \updelta_{Z_i}$ is the empirical distribution of $m$ \iid~samples generated by $\mu_{\theta}$, $\theta \in \Theta$. We assume that \begin{enumerate*}[label=(\alph*)]
\item  $\msy$, endowed with the Euclidean distance $\rho$, is a Polish space (\ie, complete and separable),
\item  $\Theta$, endowed with the distance $\rho_\Theta$, is a Polish space,
\item parameters are identifiable, \ie~$\mu_\theta = \mu_{\theta'}$ implies $\theta = \theta'$.
\end{enumerate*} $\mcb{\msy}$ denotes the Borel $\sigma$-field of $(\msy,\rho)$. \vspace{4pt}

\noindent \textbf{Approximate Bayesian Computation.} ABC methods are used to approximate the posterior distribution in generative models when the likelihood is numerically intractable but easy to sample from. The basic and simplest ABC algorithm is an acceptance-rejection method \cite{Tavare97}, which iteratively draws a candidate parameter $\theta'$ from a prior distribution $\pi$, and `synthetic data' $\bfz_{1:m} = (\bfz_i)_{i=1}^m$ from $\mu_{\theta'}$, and keeps $\theta'$ if $\bfz_{1:m}$ is close enough to the observations $\bfy_{1:n} = (\bfy_i)_{i=1}^n$. Specifically, the acceptance rule is $\bfD\big(s(\bfy_{1:n}), s(\bfz_{1:m})\big) \leq \varepsilon$, where $\bfD$ is a data discrepancy measure taking non-negative values, $\varepsilon$ is a tolerance threshold, and $s : \sqcup_{n \in \nsets} \msy^{n} \rightarrow \rset^{d_s}$ with small $d_s$ is a summary statistics. The algorithm is summarized in Pseudo-code~\ref{ps:vanilla_abc} and returns samples of $\theta$ that are distributed from:
\begin{equation} \label{eq:abc_posterior}
	\pi_{\bfy_{1:n}}^\varepsilon(\theta) = \frac{\pi(\theta) \int_{\msy^m} \1 \{ \bfD\big(s(\bfy_{1:n}), s(\bfz_{1:m})\big) \leq \varepsilon \} \rmd \mu_\theta(\bfz_{1:m})}{\int_\Theta \rmd \pi(\theta) \int_{\msy^m} \1 \{ \bfD\big(s(\bfy_{1:n}), s(\bfz_{1:m})\big) \leq \varepsilon \} \rmd \mu_\theta(\bfz_{1:m})}
\end{equation}

\noindent The choice of $s(\cdot)$ directly impacts the quality of the resulting approximate posterior: if the statistics are sufficient statistics, $\pi_{\bfy_{1:n}}^\varepsilon(\theta)$ converges to the true posterior $\pi(\theta | \bfy_{1:n})$ as $\varepsilon \to 0$, otherwise, the limiting distribution is at best $\pi(\theta | s(\bfy_{1:n}))$ \cite{Sisson2018, Frazier2018}. Wasserstein-ABC has been proposed to avoid this loss of information. \vspace{4pt}

\RestyleAlgo{ruled}
\begin{pseudocode}[t]
\caption{Vanilla ABC.}  \label{ps:vanilla_abc}
 \KwIn{observations $\bfy_{1:n}$, number of iterations $T$, data discrepancy measure $\bfD$, summary statistics $s$, tolerance threshold $\varepsilon > 0$.}
 \For{$t=1,\dots,T$}{
 \Repeat{$\bfD\big(s(\bfy_{1:n}), s(\bfz_{1:m})\big) \leq \varepsilon$}{$\theta \sim \pi(\cdot)$ and $\bfz_{1:m} \sim \mu_\theta$ i.i.d.}
 $\theta^{(t)} = \theta$
 }
 \Return $\theta^{(1)}, \dots, \theta^{(T)}$ 
\end{pseudocode}

\noindent \textbf{Wasserstein distance and ABC. } For $p \geq 1$, consider $\calP_p(\msy) = \{ \mu \in \mcp(\msy) \, : \, \int_{\msy} \norm{y - y_0 }^p \rmd\mu(y) < \plusinfty\}$ for some $y_0 \in \msy$ and the Wasserstein distance of order $p$
defined for any $\mu, \nu \in \calP_p(\msy)$ by,
\begin{equation}
\label{eq:def_wasser}
 \wassersteinD[p]^{p}(\mu, \nu) = \inf_{\gamma \in \Gamma(\mu, \nu)} \defEns{ \int_{\msy \times \msy} \norm{x - y}^p \rmd\gamma(x,y)}  \eqsp,
 \end{equation}
where $\Gamma(\mu, \nu)$ is the set of probability measures $\gamma$ on $(\msy \times \msy, \mcb{\msy} \otimes \mcb{\msy})$ verifying: $\forall \msa \in \mcb{\msy}$, $\gamma(\msa \times \msy) = \mu(\msa)$, $\gamma(\msy\times \msa) = \nu(\msa)$. %

Evaluating the Wasserstein distance between multi-dimensional probability measures turns out to be numerically intractable in general, and solving \eqref{eq:def_wasser} between empirical distributions over $n$ samples leads to computational costs in $\calO(n^3 \log(n))$ \cite{peyre2019computational}. Nevertheless, $\wassersteinD[p]$ between one-dimensional measures $\mu, \nu \in \calP_p(\rset)$ has a closed-form expression \cite[Theorem 3.1.2.(a)]{rachev:ruschendorf:1998}, given by:
\begin{align}
  \label{eq:wp1d}
    \wassersteinD[p]^p(\mu, \nu) &= \int_{0}^1 \abs{ F_\mu^{-1}(t) - F_\nu^{-1}(t) }^p \rmd t \eqsp, %
\end{align}
where %
$F_\mu^{-1}$ and $F_\nu^{-1}$ denote the quantile functions of $\mu$ and $\nu$ respectively. For empirical one-dimensional distributions, \eqref{eq:wp1d} can be efficiently approximated by simply sorting the $n$ samples drawn from each distribution and computing the average cost between the sorted samples. This amounts to $\calO(n\log(n))$ operations at worst.

Wasserstein-ABC \cite{Bernton2019} is a variant of ABC \eqref{ps:vanilla_abc} that uses $\wassersteinD[p]$, $p \geq 1$ between the empirical distributions of the observed and synthetic data, in place of the discrepancy measure $\bfD$ between summaries. To make this method scalable to any dataset size, \cite{Bernton2019} introduces a new approximation of \eqref{eq:def_wasser}, the Hilbert distance, which extends the idea behind the computation of $\wassersteinD[p]$ in 1D to higher dimensions, by sorting samples according to their projection obtained via the Hilbert space-filling curve. This alternative can be computed in $\calO(n\log(n))$, but yields accurate approximations only for low dimensions \cite{Bernton2019}. They also use a second approximation, the swapping distance, based on an iterative greedy swapping algorithm. However, each iteration requires $n^2$ operations, and there is no guarantee of convergence to $\wassersteinD[p]$.  

\section{Sliced-Wasserstein ABC}
\label{sec:swabc}
\vspace{-5pt}

\noindent \textbf{Sliced-Wasserstein distance. } The analytical expression of $\wassersteinD[p]$ in \eqref{eq:wp1d} motivates the formulation of an alternative OT distance, called the Sliced-Wasserstein distance \cite{rabin:et:al:2011, bonneel2015sliced}. SW is obtained by reducing multi-dimensional distributions to one-dimensional representations through linear projections, and then by averaging 1D Wasserstein distances between these projected distributions. More formally, we denote by $\sphereD = \set{\us \in \rset^d}{\norm{\us}_2 = 1}$ the $d$-dimensional unit sphere, and by $\ps{\cdot}{\cdot}$ the Euclidean inner-product. For any $\us \in \sphereD$, $\uss$ is the linear form associated with $\us$ such that for any $ y \in \msy,\ \uss(y) = \ps{u}{y}$. For $p \geq 1$, the Sliced-Wasserstein distance of order $p$ between $\mu,\nu \in \mathcal{P}_p(\msy)$ is defined as,
\begin{equation}
\label{eq:def_sliced_wasser}
  \swassersteinD[p]^{p}(\mu, \nu) = \int_{\sphere^{d-1}} \wassersteinD[p]^p(\uss_{\sharp} \mu, \uss_{\sharp} \nu) \rmd\unifS(\us) \eqsp ,
\end{equation}
where $\unifS$ is the uniform distribution on $\sphere^{d-1}$ and for any measurable function $f :\msy \to \rset$ and $\zeta \in \mcp(\msy)$,  $f_{\sharp}\zeta$ is the push-forward measure of $\zeta$ by $f$, defined as: $\forall \msa \in \mcb{\rset}$, $f_{\sharp}\zeta(\msa) = \zeta(f^{-1}(\msa))$, with $f^{-1}(\msa) = \{y \in \msy \, : \, f(y) \in \msa\}$. $\swassersteinD[p]$ is a distance on $\mcp_p(\msy)$ \cite{Bonnotte2013} with significantly lower computational requirements than the Wasserstein distance: in practice, the integration in \eqref{eq:def_sliced_wasser} is approximated with a finite-sample average, using a simple Monte Carlo (MC) scheme. %

SW also seems to have better statistical properties than the Wasserstein distance and its approximations. We illustrate it with the task of estimating the scaling factor of the covariance matrix in a multivariate Normal model, as in the supplementary material of \cite{Bernton2019}. For any $\sigma >0$, denote by $\mu_\sigma$ the $d$-dimensional Gaussian distribution with zero-mean and covariance matrix $\sigma^2 \bfI_d$.
Observations are assumed i.i.d. from $\mu_{\sigma_{\star}}$ with $\sigma_\star^2 = 4$, and we draw the same number of i.i.d. data from $\mu_{\sigma}$ for 100 values of $\sigma^2$ equispaced between 0.1 and 9. We then compute $\wassersteinD[2]$ and $\swassersteinD[2]$ between the empirical distributions of the samples, and the swapping and Hilbert approximations presented in \cite{Bernton2019}, for $d \in \{ 2, 10, 100 \} $ and 1000 observations. $\wassersteinD[2]$ between two Gaussian measures has an analytical formula, which boils down in our setting to: $	\wassersteinDLigne[2]^2(\mu_{\sigma_\star}, \mu_\sigma) = d(\sigma_\star - \sigma)^2$, and we approximate the exact SW using a MC approximation of: 
\begin{equation*}
	\swassersteinD[2]^2(\mu_{\sigma_\star}, \mu_\sigma) = \wassersteinD[2]^2(\mu_{\sigma_\star}, \mu_\sigma) \int_{\sphereD} u^Tu\ \rmd\unifS(\us)\eqsp, \label{eq:sw_gauss}
      \end{equation*}
This formula is derived from \eqref{eq:def_sliced_wasser} and the exact $\wassersteinD[2]$ between one-dimensional Gaussian distributions. We also compute KL with the estimator proposed for KL-based ABC (KL-ABC, \cite{Jiang18}). \Cref{fig:comparison_distances} shows the distances plotted against $\sigma^2$ for each $d$. When the dimension increases, we observe that (i) as pointed out in \cite{Bernton2019}, the quality of the approximation of empirical Wasserstein returned by Hilbert and swapping rapidly deteriorates, and (ii) SW, approximated using densities or samples, is the only approximate metric that attains its minimum at $\sigma_\star^2$. This curse of dimensionality can be a limiting factor for the performance of WABC and KL-ABC in high dimensions. \vspace{4pt}

\noindent \textbf{Sliced-Wasserstein ABC. } Motivated by the practical success of SW regardless of the dimension value in the previous experiment, we propose a variant of ABC based on SW, referred to as the Sliced-Wasserstein ABC (SW-ABC). Our method is similar to WABC in the sense that it compares empirical distributions, but instead of $\wassersteinD[p]$, we choose the discrepancy measure to be $\swassersteinD[p]$, $p \geq 1$. The usage of SW allows the method to scale better to the data size and dimension. The resulting posterior distribution, called the SW-ABC posterior, is thus defined in \eqref{eq:abc_posterior} with $\bfD$ replaced by $\swassersteinD[p]$.

\begin{figure}[t]
  \includegraphics[width=0.485\textwidth, trim={2.4mm 0 0 0}, clip]{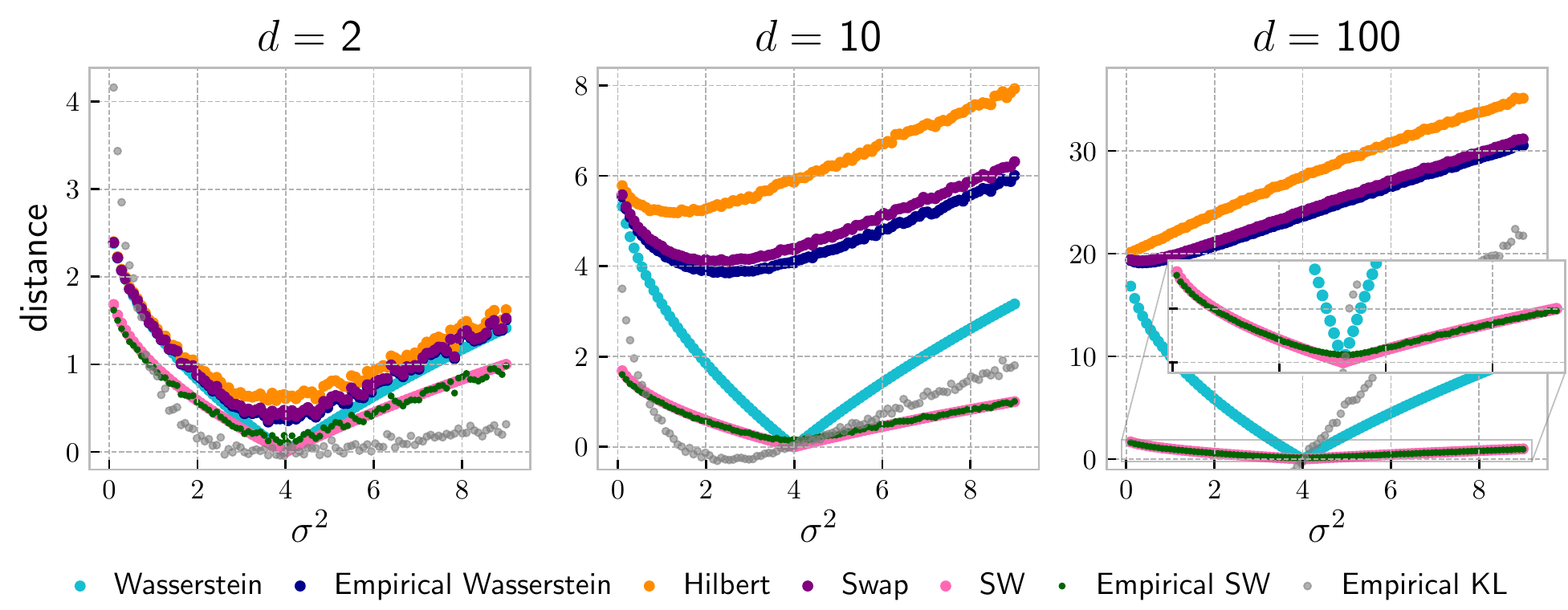}
  \caption{Comparison of OT distances and KL between data generated from $d$-dimensional Gaussian distributions $\mu_{\sigma}$ vs. $\mu_{\sigma_{\star}}$, $\sigma_{\star}^2 =4$, with 1000 i.i.d draws. SW is approximated with 100 random projections.}
  \label{fig:comparison_distances}
  \vspace{-10pt}
\end{figure}

\section{Theoretical Study}

\vspace{-5pt}

In this section, we analyze the asymptotic behavior of the SW-ABC posterior under two different regimes. Our first result concerns the situation where the observations $\bfy_{1:n}$ are fixed, and $\varepsilon$ goes to zero. We prove that the SW-ABC posterior is asymptotically consistent in the sense that it converges to the true posterior, under specific assumptions on the density used to generate synthetic data.

\begin{proposition}%
\label{prop:1} 
  Let $p \geq 1$. Suppose that $\mu_\theta$ has a density $f_\theta$ \wrt~the Lebesgue measure such that $f_\theta$ is continuous and there exists  $\calN_\Theta \subset \Theta$ satisfying $ \sup_{\theta \in \Theta \backslash \calN_\Theta} f_\theta(\bfy_{1:n}) < \infty  $ and $\pi(\calN_\Theta) = 0$. In addition, assume that there exists $\bar{\varepsilon} > 0$ such that  $ \sup_{\theta \in \Theta \backslash \calN_\Theta}\ \sup_{\bfz_{1:m} \in \msa^{\bar{\varepsilon}}} f_\theta(\bfz_{1:m}) < \infty$, where  $\msa^{\bar{\varepsilon}} = \set{ \bfz_{1:m} }{ \swassersteinD[p](\bfy_{1:n}, \bfz_{1:m}) \leq \bar{\varepsilon}}  $.
 Then, with $\bfy_{1:n}$ fixed, the SW-ABC posterior converges to the true posterior as $\varepsilon$ goes to 0, in the sense that, for any measurable $\msb \subset \Theta$, $			\lim_{\varepsilon \rightarrow 0} \pi^{\varepsilon}_{\bfy_{1:n}}(\msb) =  \pi (\msb | \bfy_{1:n})$, where $\pi^{\varepsilon}_{\bfy_{1:n}}$ is defined by \eqref{eq:abc_posterior}.
\end{proposition}

\begin{proof}[Proof of \Cref{prop:1}]
	The proof consists in applying \cite[Proposition 3.1]{Bernton2019}, which establishes the conditions for the data discrepancy measure to yield an ABC posterior that converges to the true posterior in the asymptotic regime we consider. This amounts to verify that:
	\begin{enumerate}[wide, labelwidth=!, labelindent=0pt,label=(\roman*)]
		\item For any $\bfy_{1:n}$ and $\bfz_{1:m}$, with respective empirical distributions $\hmu_n$ and $\hmu_{\theta, m}$, $\swassersteinD[p](\hmu_n, \hmu_{\theta, m}) = 0$ if and only if $\hmu_n = \hmu_{\theta, m}$. \label{cond:i}
		\item $\swassersteinD[p]$ is continuous in the sense that, if $\sequencek{\bfz^k_{1:m}}$ converges to $\bfz_{1:m}$ in the metric $\rho$, then, for any empirical distribution $\hmu_n$,\; $\lim_{k \rightarrow \infty} \swassersteinD[p](\hmu_n, \hmu^k_{\theta, m}) = \swassersteinD[p](\hmu_n, \hmu_{\theta, m})$, where $\hmu^k_{\theta, m}$ is the empirical measure of $\bfz_{1:m}^k$. \label{cond:ii}
	\end{enumerate}

	\noindent Condition~\ref{cond:i} follows from the fact that $\swassersteinD[p]$ is a distance \cite[Proposition 5.1.2]{Bonnotte2013}.
	Now, let $\bfy' \in \msy$ and $\psi : \msy \to \rset$ be a continuous function such that for any $ \bfy \in \msy,\ | \psi(\bfy) | \leq K \big( 1 + \rho(\bfy', \bfy)^p \big)$ with $K \in \rset$. Since $\sequencek{\bfz^k_{1:m}}$ converges to $\bfz_{1:m}$ in the metric $\rho$ and $\psi$ is continuous, we get that $ \lim_{k \rightarrow \infty} \int \psi\ \rmd \hmu_{\theta, m}^k =  \int \psi\ \rmd \hmu_{\theta, m} $.
	This implies that $\hmu_{\theta, m}^k$ weakly converges to $\hmu_{\theta, m}$ in $\calP_p(\msy)$ \cite[Definition 6.7]{Villani2008}, which, by \cite[Theorem 6.8]{Villani2008}, is equivalent to $\lim_{k \rightarrow \infty} \wassersteinD[p](\hmu_{\theta, m}^k, \hmu_{\theta, m}) = 0$. By applying the triangle inequality and \cite[Proposition 5.1.3]{Bonnotte2013}, there exists $C \geq 0$ such that, for any empirical measure $\hmu_n$,
    \begin{align*}
        | \swassersteinD[p](\hmu_n, \hmu_{\theta, m}^k) - \swassersteinD[p](\hmu_n, \hmu_{\theta, m}) | &\leq \swassersteinD[p](\hmu_{\theta, m}^k, \hmu_{\theta, m}) \\
        &\leq C^{1/p}\  \wassersteinD[p](\hmu_{\theta, m}^k, \hmu_{\theta, m}) \eqsp .
    \end{align*}
    \noindent We conclude that $\lim_{k \rightarrow \infty} \swassersteinD[p](\hmu_n, \hmu_{\theta, m}^k) = \swassersteinD[p](\hmu_n, \hmu_{\theta, m})$, making condition~\ref{cond:ii} applicable. 
\end{proof}

Next, we study the limiting SW-ABC posterior when the value of $\varepsilon$ is fixed and the number of observations increases, \ie~$n \rightarrow \infty$. We suppose the size $m$ of the synthetic dataset grows to $\alpha n$ with $\alpha > 0$, such that $m$ can be written as a function of $n$, $m(n)$, satisfying $\lim_{n \rightarrow \infty} m(n) = \infty$.  We show that, under this setting and appropriate additional conditions, the resulting approximate posterior converges to the prior distribution on $\theta$ restricted to the region $\set{\theta \in \Theta}{\swassersteinD[p](\mu_{\theta_\star}, \mu_\theta) \leq \varepsilon}$. %

\begin{proposition}%
\label{prop:2} 
  Let $p\geq 1$, $\varepsilon >0$ and  $(m(n))_{n \in \nset^*}$ be an increasing sequence satisfying $\lim_{n \rightarrow \infty} m(n)/n = \alpha$, for $\alpha >0$. Assume that the statistical model $\model_{\Theta}$ is well specified, \ie~there exists $\theta_\star \in \Theta$ such that $\mu_\star = \mu_{\theta_\star}$,  and that almost surely the following holds:
\begin{equation}
  \label{assum:cvg_wass}
  \lim_{n \rightarrow \infty} \swassersteinD[p](\hmu_n, \hmu_{\theta, m(n)}) = \swassersteinD[p](\mu_{\theta_\star}, \mu_\theta) \eqsp,
\end{equation}
where $\hmu_n$, $\hmu_{\theta, m(n)}$ denote the empirical distributions of the observations $\bfy_{1:n}$ and synthetic data $\bfz_{1:m(n)}$ respectively. Then, the SW-ABC posterior converges to the restriction of the prior $\pi$ on the region $\set{\theta \in \Theta}{\swassersteinD[p](\mu_{\theta_\star}, \mu_\theta) \leq \varepsilon}$ as $n \rightarrow \infty$, %
\ie~for any $\theta \in \Theta$,
	\begin{align*}
		\lim_{n \rightarrow \infty} \pi_{\bfy_{1:n}}^{\varepsilon}(\theta) &= \pi(\theta | \swassersteinD[p](\mu_{\theta_\star}, \mu_\theta) \leq \varepsilon) \\
		&\propto \pi(\theta) \1 \{ \swassersteinD[p](\mu_{\theta_\star}, \mu_\theta) \leq \varepsilon \} \eqsp .
	\end{align*}
\end{proposition}

\begin{proof}[Proof of \Cref{prop:2}]
	The result follows from the application of \cite[Theorem 1]{Jiang18} to $\swassersteinD[p]$ and the required conditions. %
\end{proof}

\begin{remark}
	Condition \eqref{assum:cvg_wass} is a mild assumption, e.g. is fulfilled if $\msy$ is compact and separable: in this case, for any $\nu \in \calP_p(\msy)$ and its empirical instantiation $\hnu_n$ %
	, $\lim_{n \rightarrow \infty} \wassersteinD[p](\nu, \hnu_n) = 0$ $\nu$-almost surely \cite{Weed2019}, then $\lim_{n \rightarrow \infty} \swassersteinD[p](\nu, \hnu_n) = 0$ $\nu$-almost surely \cite[Proposition 5.1.3]{Bonnotte2013}, and \eqref{assum:cvg_wass} follows by applying the triangle inequality.
\end{remark}

\section{Experiments}
\label{sec:experiments}

\vspace{-5pt}

\begin{figure}[t] 
  \begin{minipage}[b]{\linewidth}
    \centering
    \centerline{\includegraphics[width=\textwidth]{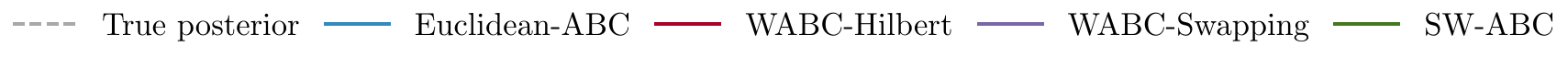}}
    \centerline{}\medskip \vspace{-0.6cm}
  \end{minipage}
  \begin{minipage}[b]{.3465\linewidth}
    \centering
    \centerline{\includegraphics[width=\textwidth, trim={2mm 3mm 1mm 2mm}, clip]{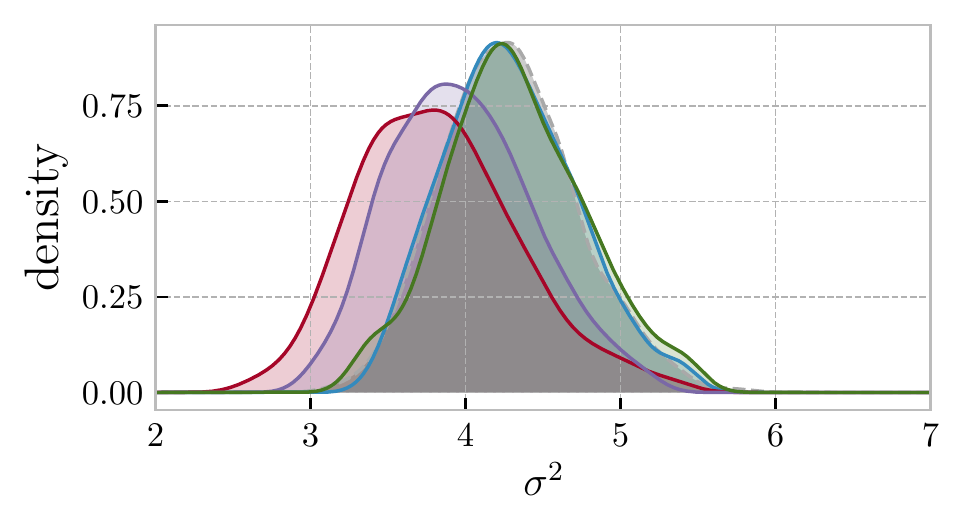}}
    \centerline{}
  \end{minipage}
  \hspace{-.5mm}
  \begin{minipage}[b]{.319\linewidth}
    \centering
    \centerline{\includegraphics[width=\textwidth, trim={8mm 3mm 1mm 2mm}, clip]{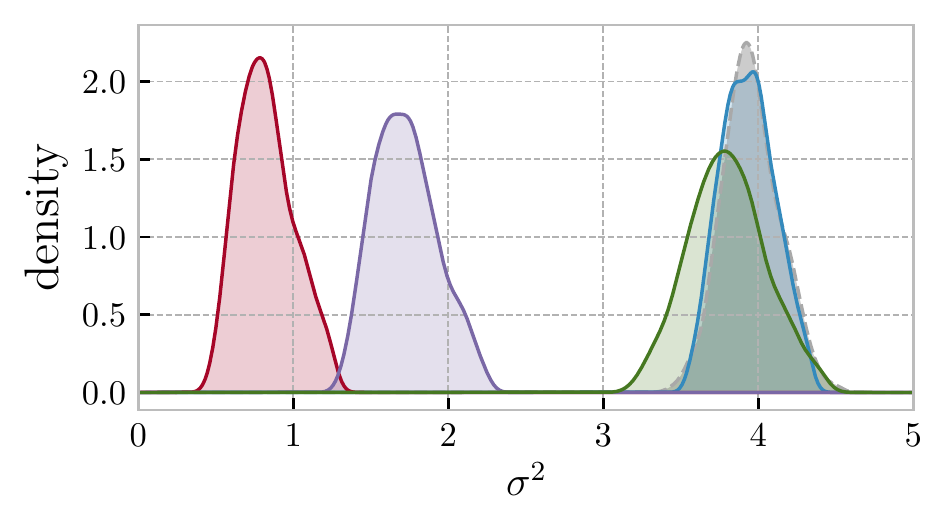}}
    \centerline{}
  \end{minipage}
  \hspace{-.5mm}
  \begin{minipage}[b]{.309\linewidth}
    \centering
    \centerline{\includegraphics[width=\textwidth, trim={8mm 3mm 1mm 2mm}, clip]{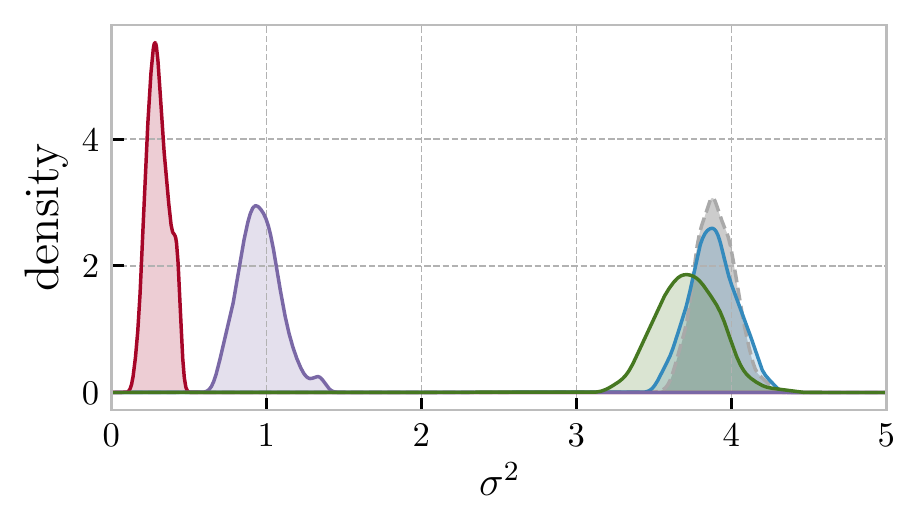}}
    \centerline{}
  \end{minipage}  \vspace{-6.5mm} \\
  \begin{minipage}[b]{.342\linewidth}
    \centering
    \centerline{\includegraphics[width=\textwidth, trim={2mm 3mm 2mm 2mm}, clip]{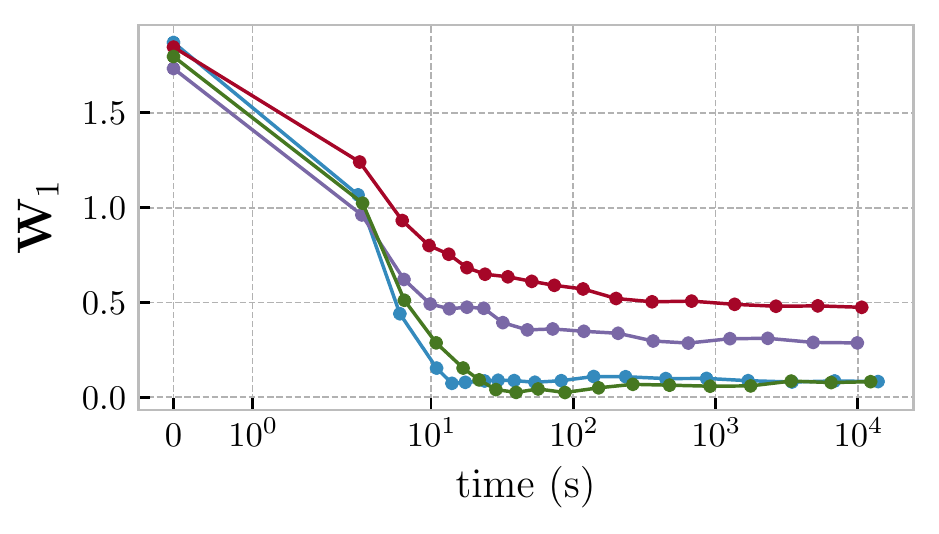}}
    \centerline{(a) $d = 2$}
  \end{minipage}
  \hspace{0mm}
  \begin{minipage}[b]{.313\linewidth}
    \centering
    \centerline{\includegraphics[width=\textwidth, trim={7mm 3mm 2mm 2mm}, clip]{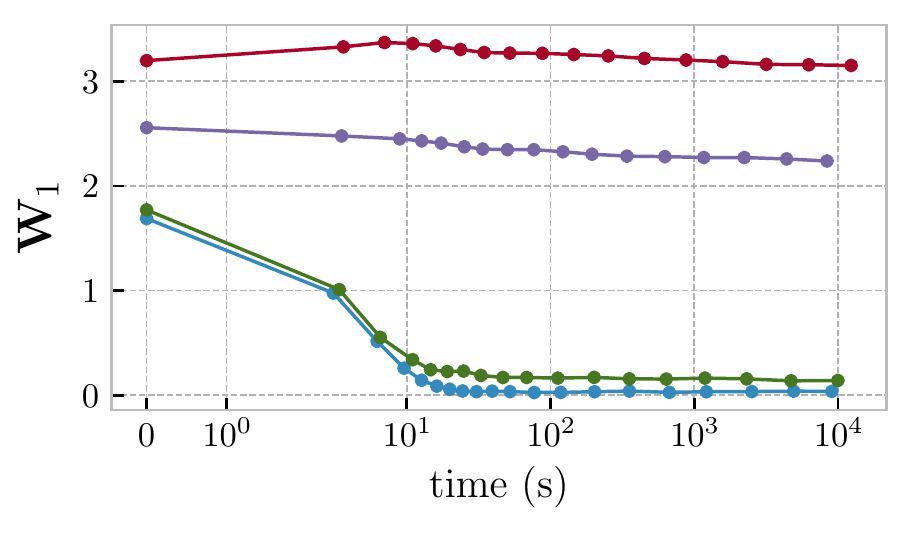}}
    \centerline{(b) $d = 10$} 
  \end{minipage}
  \hspace{0mm}
  \begin{minipage}[b]{.313\linewidth}
    \centering
    \centerline{\includegraphics[width=\textwidth, trim={7mm 3mm 2mm 2mm}, clip]{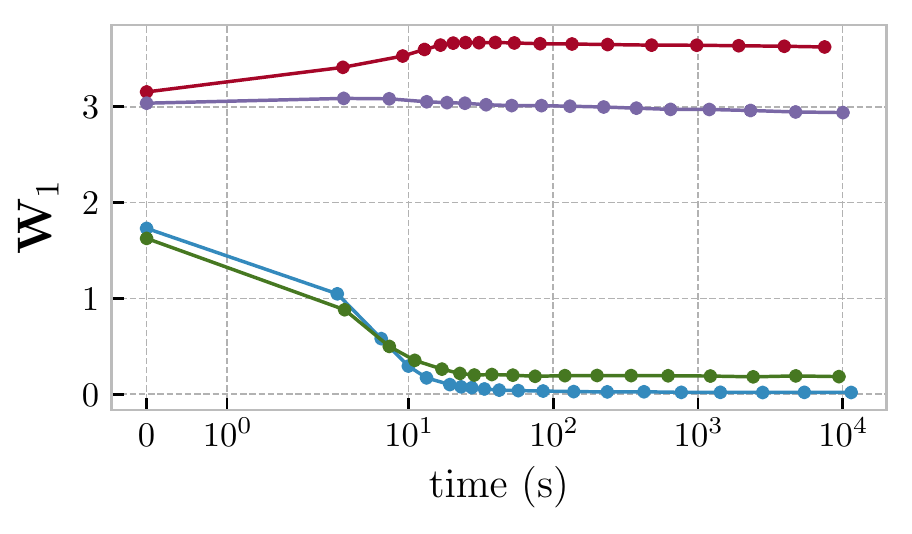}}
    \centerline{(c) $d = 20$}
  \end{minipage}
  \vspace{-11pt}
  \caption{Comparison of SMC-ABC strategies in the multivariate Gaussians problem. Each strategy uses 1000 particles and are run for 3 hours max. First row shows ABC and true posteriors of $\sigma^2$, second row reports $\wassersteinD[1]$-distance to true posterior vs. time. SW is approximated with its MC estimate over 100 random projections.}
  \label{fig:gaussian_exp}
  \vspace{-10pt}
\end{figure}

\noindent \textbf{Multivariate Gaussians.} As a first set of experiments, we investigate the performance of SW-ABC on a synthetical setting where the posterior distribution is analytically available. We consider $n = 100$ observations $(\bfy_i)_{i=1}^n$ \iid~from a $d$-dimensional Gaussian $\calN(\bfm_\star, \sigma_\star^2 \bfI_d)$, with $\bfm_\star \sim \calN(\bfnot, \bfI_d)$ and $\sigma_\star^2 = 4$. The parameter $\theta$ is $\sigma^2$ for which the prior distribution is assigned to be an inverse gamma distribution $\mathcal{IG}(1, 1)$. Therefore, the posterior distribution of $\sigma^2$ given $(\bfy_i)_{i=1}^n$ and $\bfm_{\star}$ is $\mathcal{IG}(1 + n \cdot d/2, 1 + 2^{-1}\sum_{i=1}^n \norm[2]{\bfy_i - \bfm_\star}) $. %
We compare SW-ABC against ABC using the Euclidean distance between sample variances (Euclidean-ABC), WABC with the Hilbert distance, WABC with the swapping distance and KL-ABC. Each ABC approximation was obtained using the sequential Monte Carlo sampler-based ABC method \cite{Toni2009}, which is more computationally efficient than vanilla ABC~\eqref{ps:vanilla_abc} and implemented in the package \texttt{pyABC} \cite{Klinger2018}. We provide our code in \cite{swabc_code}. \Cref{fig:gaussian_exp} reports for $d \in \{2, 10, 20\}$, the resulting ABC posteriors and $\wassersteinD[1]$ to the true posterior vs. time. Due to the poor performance of the estimator of KL between two empirical distributions proposed in \cite{Jiang18} (see Fig.~\ref{fig:comparison_distances}), KL-ABC fails at approximating well the posterior in these experiments. Hence, we excluded those results from Fig.~\ref{fig:gaussian_exp} for clarity. Euclidean-ABC provides the most accurate approximation as expected since it relies on statistics that are sufficient in our setting. %
WABC performs poorly with high-dimensional observations, contrary to SW-ABC, which approximates well the posterior for each dimension value and is as fast. \vspace{4pt}

\noindent \textbf{Image denoising.} %
We now evaluate our approach on a real application, namely image denoising. We consider a widely used algorithm for this task, the Non-Local Means algorithm (NL-means, \cite{buades2005non}), and we present a novel variant of it derived from SW-ABC. 

We formally define the denoising problem: let $\clean \in \rset^{M \times N}$, denote a clean gray-level image. We observe a corrupted version of this image, $\noisy = \clean + \noise$, where $\noise$ is some noise in $\rset^{M \times N}$. %
The goal is to restore $\clean$ from $\noisy$. We focus on denoising methods that consider `patch-based representations' of images, e.g. NL-means. Specifically, let $r \in \mathbb{N}$ be a patch size and $\indiceset = \{1, \dots, M\} \times \{1, \dots, N\}$ the set of pixel positions. For $\indice \in \indiceset$, $\clean'(\indice)$ denotes the pixel value at position $\indice$ in image $\clean'$, and $\patch_{\indice}$ is a $(2r+1) \times (2r+1)$ window in $\noisy$ centered at $\indice$: for $\indicetrois \in \{-r, \dots, r\}^2$,\; $\patch_{\indice}(\indicetrois) = \noisy(\indice + \indicetrois)$, where $\noisy$ is extended to $\zset^2$ by periodicity. Let $\dic \subset \indiceset$ be a dictionary of positions, and $\phi : I \to D$ such that, for $\indice \in I$, $\phi(\indice) = \argmin_{\indicedic \in D} \norm{\patch_{\indice} - \patch_{\indicedic}}_2$, \ie~$\phi(\indice)$ is the position in $\dic$ of the most similar patch to $\patch_\indice$. %
For $\indicedic \in \dic$, an estimator of $\patch_{\indicedic}$ is given by $\hat{P}_{\indicedic} = \mathbb{E}_{\pi(\indice | (P_{\indicetrois})_{\indicetrois \in \phi^{-1}(\indicedic)}) \tilde{\pi}(\mathbf{l})}[ P_{\mathbf{i} + \mathbf{l}} ]$,
$\tilde{\pi}$ being the uniform distribution on $\phi^{-1}(\indicedic)$. In practice, it is approximated with a Monte Carlo scheme: 
\begin{equation}
  \label{eq:hat_p_approx}
  \hat{P}_{\indicedic} \approx (Tn)^{-1} \sum\nolimits_{t=1}^T \sum\nolimits_{s=1}^S P_{\indice^{(t)} + \mathbf{l}^{(s)}} \eqsp ,
\end{equation}
where $\indice^{(t)} \sim \pi(\indice^{(t)} | (P_{\indicetrois})_{\indicetrois \in \phi^{-1}(\indicedic)})$, $\mathbf{l}^{(s)} \sim \tilde{\pi}(\mathbf{l})$, and $\indice$, $\mathbf{l}$ are mutually independent. Finally, we construct an estimate $\hat{\clean}$ of $\clean$ as follows: for any $\indice \in \indiceset$, \; $\hat{\clean}({\indice}) =   \sum_{\indicetrois, \|\indicetrois - \indice\|_\infty \leq r  } \hat{P}_{\phi(\indicetrois)}(\indice - \indicetrois) \eqsp (2 r + 1)^{-2}$. %
The classical NL-means estimator corresponds to the case where $\dic = \indiceset$ (thus $\phi  = \Id$) and for any $\indice \in I$ and $P \in \rset^{(2r+1) \times (2r+1)}$, $\pi(\indice, P) \propto \1_{W}(\indice) \rme^{- \norm{\patch - \patch_{\indice}}^2 / (2\sigma^2)}$, where $W$ is a search window. 

In our work, we assume that the likelihood $\pi(P | \indice)$ is not available, but we observe for $j \in \dic$, %
$(P_{\indicetrois_\ell})_{\ell=1}^m$ ($\indicetrois_\ell \in  \phi^{-1}(\indicedic)$) %
\iid~from $\pi(\cdot | \indice)$.
By replacing %
$\pi(\indice | (P_{\indicetrois_\ell})_{\ell=1}^m)$ in \eqref{eq:hat_p_approx} by the SW-ABC posterior, we obtain the proposed denoising method, called the SW-ABC NL-means algorithm. We provide our Python implementation in \cite{swabc_code}.

We compare our approach with the classical NL-means. We consider one of the standard image denoising datasets \cite{fan2019brief}, called CBSD68 \cite{martin2001database} and consisting of $68$ colored images of size $321 \times 481$. We first convert the images to gray scale, then manually corrupt each of them with a Gaussian noise with standard deviation $\sigma $, and try to recover the clean image. %
The quality of the output images is evaluated with the Peak Signal to Noise Ratio ($\PSNR$) measure:\; $\PSNR = - 10 \log_{10}\defEns{ \norm{\clean - \hat{\clean}}_2^2 / (255^2 N M)}$. In our experiments, we use a dictionary of $1000$ patches picked uniformly at random, we set $T = S = m = 10$, $r = 3$, $W = \{-10, \dots, 10\}^2$, $\vareps = (2 r + 1)^2$, and we compute SW with a MC scheme over $L=100$ projections.

We report the average PSNR values for different values of the noise level $\sigma$ in Table~\ref{tbl:results}. We observe that for small $\sigma$, NL-means provides more accurate results, whereas when $\sigma$ becomes larger SW-ABC outperforms NL-means, thanks to the patch representation and the use of SW. 
On the other hand, another important advantage of SW-ABC becomes prominent in the computation time: the proposed approach takes $\approx 6s$ on a standard laptop computer per image whereas the classical NL-means algorithm takes $\approx 30s$. Indeed, the computational complexity of SW-ABC NL-means is upper-bounded by %
$\abs{D}TS\hspace{.2mm}C_{\mathrm{SW}}$, with $C_{\mathrm{SW}} = Lm\log(m)$ is the cost of computing SW, whereas it is $NM\abs{W}(2r+1)^2$ for the na\"{i}ve implementation of NL-means, where $\abs{A}$ denotes the cardinal number of the set $A$. %
We can observe that SW-ABC has a lower computational complexity since $\abs{D} \ll NM$ in practice. We note that the computation time of NL-means can be improved by certain acceleration techniques, which can be directly used to improve the speed of SW-ABC NL-means as well. Finally, in \Cref{fig:multiple_images}, we illustrate the performance of SW-ABC on two $512 \times 512$ images for visual inspection. The results show that the injected noise is successfully removed by the proposed approach.

\begin{table}
\caption{Comparison of NL-means and SW-ABC on the image denoising task in terms of average PSNR (dB). For each $\sigma$, we fine-tuned the hyperparameters of NL-means and reported the best result.}
\label{tbl:results}
\vspace{-20pt}
\begin{center}

\renewcommand{\arraystretch}{0.8}
\begin{tabular}{ l | c | c | c | c |}
\toprule
& $\sigma = 10$ & $\sigma = 20$ & $\sigma = 30$ & $\sigma = 50$ \\
\toprule
NL-means & $\mathbf{30.43}$ & $\mathbf{26.32}$ & $24.22$ & $21.99$\\ 
\midrule
SW-ABC & $27.09$ & $26.26$ & $\mathbf{24.86}$ & $\mathbf{22.56}$\\ 
\bottomrule
\end{tabular}
\end{center}
\vspace{-28pt}
\end{table}

\newcommand{\figsize}{0.2428}
\captionsetup[subfigure]{labelformat=empty}
\begin{figure}[t]
  \subfloat[]{\includegraphics[width=\figsize\linewidth]{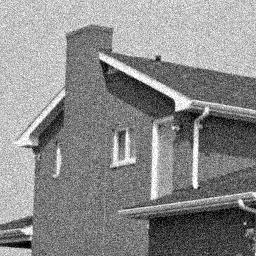}} \hspace{.2pt} 
  \subfloat[]{\includegraphics[width=\figsize\linewidth]{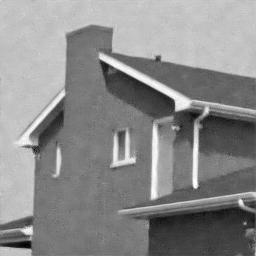}}  \hspace{1.5pt} %
  \subfloat[]{\includegraphics[width=\figsize\linewidth]{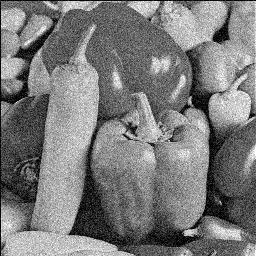}} \hspace{.2pt} 
  \subfloat[]{\includegraphics[width=\figsize\linewidth]{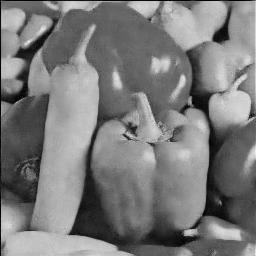}} %
  \vspace{-0.3cm}
  \vspace{-5pt}
  \caption{Visualization of the results. For each couple, the left one is the noisy image ($\sigma=20$) and the right one is the output of SW-ABC. %
  }
  \label{fig:multiple_images}
  \vspace{-15pt}
\end{figure}

\section{Conclusion}
\label{sec:conclusion}

\vspace{-5pt}

We proposed a novel ABC method, SW-ABC, based on the Sliced-Wasserstein distance. We derived asymptotic guarantees for the convergence of the SW-ABC posterior to the true posterior under different regimes, and we evaluated our approach on a synthetical and an image denoising problem. Our results showed that SW-ABC provides an accurate approximation of the posterior, even with high-dimensional data spaces and a small number of samples. Future work includes extending SW-ABC to generalized SW distances \cite{Kolouri2019GSW}.

\bibliographystyle{IEEEbib}
\bibliography{refs_swabc}

\end{document}